\documentclass[10pt,twocolumn]{article}

\usepackage[
  a4paper,
  left=1.8cm,
  right=1.8cm,
  top=1.8cm,
  bottom=2.0cm
]{geometry}
\usepackage[T1]{fontenc}
\usepackage{lmodern}
\usepackage{microtype}
\usepackage{amsmath,amsthm}
\usepackage{braket}
\usepackage{enumitem}
\usepackage{titlesec}
\usepackage[hidelinks]{hyperref}
\usepackage{orcidlink}

\setlist[itemize]{
  leftmargin=1.25em,
  itemsep=0.2em,
  topsep=0.35em,
  parsep=0pt
}

\renewcommand{\thesection}{\Roman{section}}

\titleformat{\section}[block]
  {\normalfont\normalsize\bfseries\centering}
  {\thesection.}
  {0.6em}
  {\MakeUppercase}

\titlespacing*{\section}
  {0pt}
  {1.7ex plus 0.5ex minus 0.2ex}
  {0.6ex plus 0.2ex}

\theoremstyle{plain}
\newtheorem{proposition}{Proposition}
\newtheorem{corollary}[proposition]{Corollary}

\hypersetup{
  pdftitle={
    Comment on Scalable Quantum Machine Learning:
    Trainability, Expressivity and Efficiency:
    Polynomial Evaluation of the Triplet-Block Readout
  },
  pdfauthor={Erfan Amidi}
}

\begin{document}

\twocolumn[
\begin{center}

{\LARGE\bfseries
Comment on ``Scalable Quantum Machine Learning:
Trainability, Expressivity and Efficiency'':\\[0.25em]
Polynomial Evaluation of the Triplet-Block Readout
\par}

\vspace{0.8em}

{\large
Erfan Amidi\,\orcidlink{0009-0007-2612-6313}
\par}

\vspace{0.45em}

{\normalsize\itshape
Institute for Advanced Studies in Basic Sciences
\par}

\vspace{0.85em}

\begin{minipage}{0.96\textwidth}
\small

We examine the classical-cost claim for the triplet-block two-body
readout in Ref.~\cite{Kerenidis2026v1}. The Gaussian-state expansion
used there gives an \(O(2^{2k/3}\operatorname{poly}(n))\) classical
algorithm, but it is not necessary for fixed-body observables.
The triplet-block input has an explicitly computable diagonal
two-particle reduced density matrix, which passive fermionic linear
optics propagates through \(\bigwedge^2 W\). This gives a deterministic
\(O(n^4)\) algorithm for the complete correlator vector
\((\langle n_i n_j\rangle)_{i<j}\), independently of \(k\) and the
fermionic-linear-optics extent. More generally, every
number-conserving fixed-\(r\)-body expectation is polynomially
computable whenever the input \(r\)-particle reduced density matrix is
classically available; if that matrix is diagonal, all diagonal
correlators are computable in \(O(n^{2r})\) time. This invalidates the
algorithm-relative exponential-cost conclusion for the supervised
two-body readout, without affecting the gradient-variance,
barren-plateau, parameter-shift, or sampling-hardness results.
\end{minipage}

\end{center}

\vspace{0.8em}
]

\section{Claim under examination}

Reference~\cite{Kerenidis2026v1} considers \(n\) fermionic modes and
the triplet-block input
\begin{equation}
\begin{gathered}
\ket{\Psi_{\mathrm{in}}}
=
\ket{\Psi_6}^{\otimes m}
\otimes
\ket{0}^{\otimes(n-6m)},
\\
\ket{\Psi_6}
=
\frac{\ket{111000}+\ket{000111}}{\sqrt{2}},
\qquad
m=\frac{k}{3}.
\end{gathered}
\label{eq:triplet-input}
\end{equation}
The fixed spreading circuit, data-dependent phase layer, and trainable
brick-wall or butterfly circuit are all passive fermionic-linear-optics
operations. Their combined action is therefore represented by a
single-particle unitary
\begin{equation}
W(\mathbf{x},\theta)
=
W_{\mathrm{train}}(\theta)D(\mathbf{x})H\in U(n).
\label{eq:total-single-particle-unitary}
\end{equation}
The supervised readout is the two-body correlator vector
\begin{equation}
\begin{gathered}
\mathbf q(\mathbf{x},\theta)
=
\bigl(q_{ij}(\mathbf{x},\theta)\bigr)_{i<j},
\\
q_{ij}
=
\langle n_i n_j\rangle,
\qquad
n_i=a_i^\dagger a_i.
\end{gathered}
\label{eq:two-body-readout}
\end{equation}
Each block in \eqref{eq:triplet-input} is a sum of two Slater
determinants, so the complete input has a Gaussian decomposition with
\(\chi=2^m=2^{k/3}\) terms:
\[
\ket{\Psi_{\mathrm{in}}}
=
\sum_{s=1}^{\chi}
c_s\ket{\phi_s^{\mathrm{Gauss}}}.
\]
Expanding a two-body expectation value in this decomposition gives
\begin{equation}
q_{ij}
=
\sum_{s,s'=1}^{\chi}
c_s^\ast c_{s'}
\bra{\phi_s^{\mathrm{Gauss}}}
\mathcal U_W^\dagger n_i n_j\mathcal U_W
\ket{\phi_{s'}^{\mathrm{Gauss}}}.
\label{eq:gaussian-double-sum}
\end{equation}
Evaluating the \(\chi^2\) cross terms separately produces the upper
bound
\[
O\!\left(
\chi^2\operatorname{poly}(n)
\right)
=
O\!\left(
2^{2k/3}\operatorname{poly}(n)
\right).
\]
Reference~\cite{Kerenidis2026v1} notes that the mixed-Pfaffian
compression available for form-rank-\(2\) inputs does not apply to the
form-rank-\(3\) triplet state, and it consequently identifies
\eqref{eq:gaussian-double-sum} as the best available route. This
argument supplies an upper bound for one algorithm, rather than a
complexity-theoretic lower bound. More importantly, the full Gaussian
expansion is unnecessary for the observable in
\eqref{eq:two-body-readout}.

\section{Fixed-body closure and polynomial algorithm}
\label{sec:fixed-body-closure}

For a fixed positive integer \(r\leq n\), define the ordered index set
\[
\mathcal I_r
:=
\{(i_1,\ldots,i_r):1\leq i_1<\cdots<i_r\leq n\}.
\]
For \(I=(i_1,\ldots,i_r)\in\mathcal I_r\), write
\[
a_I:=a_{i_r}\cdots a_{i_1},
\qquad
a_I^\dagger:=a_{i_1}^\dagger\cdots a_{i_r}^\dagger.
\]
The \(r\)-particle reduced density matrix of a fermionic state \(\rho\)
is the \(\binom nr\times\binom nr\) matrix
\begin{equation}
[\Gamma_\rho^{(r)}]_{I,J}
:=
\operatorname{Tr}\!\left[\rho\,a_J^\dagger a_I\right],
\qquad I,J\in\mathcal I_r.
\label{eq:r-rdm-definition}
\end{equation}
Its diagonal contains the fixed-body number correlators,
\begin{equation}
[\Gamma_\rho^{(r)}]_{I,I}
=
\operatorname{Tr}\!\left[
\rho\prod_{\ell=1}^r n_{i_\ell}
\right].
\label{eq:r-rdm-diagonal-correlator}
\end{equation}
Let \(\mathcal U_W\) be a passive fermionic-linear-optics
unitary~\cite{Bravyi2005}, with
\begin{equation}
\mathcal U_W^\dagger a_i\mathcal U_W
=
\sum_{p=1}^nW_{ip}a_p.
\label{eq:passive-flo-action}
\end{equation}
For \(I,P\in\mathcal I_r\), let \(W_{I,P}\) denote the \(r\times r\)
submatrix with rows indexed by \(I\) and columns indexed by \(P\).
The induced action on the \(r\)-particle sector is
\begin{equation}
R_r(W):=\bigwedge\nolimits^r W,
\qquad
[R_r(W)]_{I,P}:=\det W_{I,P}.
\label{eq:exterior-power-matrix}
\end{equation}

\begin{proposition}[Fixed-body closure under passive fermionic linear optics]
\label{prop:fixed-body-closure}
For every fermionic state \(\rho\), passive-FLO unitary
\(\mathcal U_W\), and \(1\leq r\leq n\),
\begin{equation}
\Gamma_{\mathcal U_W\rho\mathcal U_W^\dagger}^{(r)}
=
R_r(W)\Gamma_\rho^{(r)}R_r(W)^\dagger.
\label{eq:r-rdm-propagation}
\end{equation}
Consequently, every number-conserving \(r\)-body observable is
determined by \(\Gamma_\rho^{(r)}\). Writing
\[
N_r
:=
|\mathcal I_r|
=
\binom nr
=
O(n^r)
\qquad
\text{for fixed }r,
\]
the matrices in \eqref{eq:r-rdm-propagation} have dimension
\(N_r\times N_r\). If the input \(r\)-particle reduced density matrix
is classically available, constructing \(R_r(W)\) and performing the
dense propagation require
\[
O(r^3N_r^2)
\qquad\text{and}\qquad
O(N_r^3)
\]
arithmetic operations, respectively. Hence every number-conserving
fixed-\(r\)-body expectation is classically computable in polynomial
time, independently of the form-rank of the full input state.
\end{proposition}

\begin{proof}
Applying \eqref{eq:passive-flo-action} to the annihilation string gives
\[
\mathcal U_W^\dagger a_I\mathcal U_W
=
\left(
\sum_{p_r=1}^n W_{i_rp_r}a_{p_r}
\right)
\cdots
\left(
\sum_{p_1=1}^n W_{i_1p_1}a_{p_1}
\right).
\]
Any term containing a repeated mode vanishes because \(a_p^2=0\).
For a fixed ordered set
\(P=(p_1<\cdots<p_r)\), collect the \(r!\) permutations of its modes
and reorder the annihilation operators using
\(a_pa_q=-a_qa_p\). The resulting coefficient is
\[
\sum_{\sigma\in S_r}
\operatorname{sgn}(\sigma)
\prod_{\ell=1}^r
W_{i_\ell p_{\sigma(\ell)}}
=
\det W_{I,P}.
\]
Therefore,
\begin{equation}
\mathcal U_W^\dagger a_I\mathcal U_W
=
\sum_{P\in\mathcal I_r}
[R_r(W)]_{I,P}a_P.
\label{eq:annihilation-string-transform}
\end{equation}
Taking adjoints similarly gives
\begin{equation}
\mathcal U_W^\dagger a_J^\dagger\mathcal U_W
=
\sum_{Q\in\mathcal I_r}
\overline{[R_r(W)]_{J,Q}}a_Q^\dagger.
\label{eq:creation-string-transform}
\end{equation}
Let
\(\rho_{\mathrm{out}}
=\mathcal U_W\rho\mathcal U_W^\dagger\).
Cyclicity of the trace and
\eqref{eq:annihilation-string-transform}--%
\eqref{eq:creation-string-transform} imply
\begin{align*}
[\Gamma_{\rho_{\mathrm{out}}}^{(r)}]_{I,J}
&=
\operatorname{Tr}\!\left[
\rho\,
\mathcal U_W^\dagger
a_J^\dagger a_I
\mathcal U_W
\right]
\\
&=
\sum_{P,Q\in\mathcal I_r}
[R_r(W)]_{I,P}
[\Gamma_\rho^{(r)}]_{P,Q}
\overline{[R_r(W)]_{J,Q}}.
\end{align*}
This is precisely the matrix identity
\eqref{eq:r-rdm-propagation}.
Every number-conserving \(r\)-body observable can be written as
\[
O^{(r)}
=
\sum_{I,J\in\mathcal I_r}
o_{I,J}a_J^\dagger a_I.
\]
By the definition of the \(r\)-particle reduced density matrix,
\[
\operatorname{Tr}\!\left(\rho O^{(r)}\right)
=
\sum_{I,J\in\mathcal I_r}
o_{I,J}
[\Gamma_\rho^{(r)}]_{I,J}.
\]
Hence every such expectation value is determined entirely by
\(\Gamma_\rho^{(r)}\).

The matrix \(R_r(W)\) has \(N_r^2\) entries, each given by a
determinant of size \(r\). Since each determinant is computable in
\(O(r^3)\) arithmetic operations, constructing the complete matrix
costs \(O(r^3N_r^2)\).
The propagation in \eqref{eq:r-rdm-propagation} involves ordinary
multiplication of dense \(N_r\times N_r\) matrices and therefore costs
\(O(N_r^3)\) arithmetic operations. Since \(N_r=O(n^r)\), both costs
are polynomial in \(n\) whenever \(r\) is fixed.
No step of the derivation uses a Gaussian decomposition or assumes any
bound on the form-rank of the input state.
\end{proof}

\begin{corollary}[Diagonal input reduced density matrix]
\label{cor:diagonal-r-rdm}
Suppose
\(\Gamma_\rho^{(r)}=\operatorname{diag}(w_P)_{P\in\mathcal I_r}\).
For every \(I\in\mathcal I_r\), the output diagonal correlator is
\begin{equation}
q_I
:=
\operatorname{Tr}\!\left[
\mathcal U_W\rho\mathcal U_W^\dagger
\prod_{\ell=1}^r n_{i_\ell}
\right]
=
\sum_{P\in\mathcal I_r}
w_P\left|\det W_{I,P}\right|^2.
\label{eq:diagonal-r-body-algorithm}
\end{equation}
For fixed \(r\), all \(\binom nr\) diagonal correlators are therefore
computable in
\(O(r^3\binom nr^2)=O(n^{2r})\) arithmetic operations.
\end{corollary}

\begin{proof}
Taking the \((I,I)\) entry of
\eqref{eq:r-rdm-propagation} and using diagonality gives
\eqref{eq:diagonal-r-body-algorithm}. There are \(\binom nr\) output
indices, \(\binom nr\) summands for each, and each \(r\times r\)
determinant costs \(O(r^3)\) arithmetic operations.
\end{proof}

The fixed-body assumption is essential for the stated complexity:
if \(r\) grows with \(n\), then \(\binom nr\) need not be polynomial.
Likewise, Proposition~\ref{prop:fixed-body-closure} is a closure
statement, not by itself an input-simulability theorem; classical
evaluation requires the input \(r\)-particle reduced density matrix to
be classically available.
For the triplet input in \eqref{eq:triplet-input}, set \(r=2\).
For block \(b\), let \(A_b\) and \(B_b\) denote its two disjoint
three-mode branches. Then
\[
\ket{\Psi_6^{(b)}}
=
\frac{\ket{A_b}+\ket{B_b}}{\sqrt2}.
\]
A two-body monomial can replace at most two occupied modes, whereas
\(\ket{A_b}\) and \(\ket{B_b}\) differ in three occupied modes.
Therefore
\[
\bra{A_b}
a_u^\dagger a_v^\dagger a_y a_x
\ket{B_b}
=
0
\]
for all mode indices \(u,v,x,y\). Distinct global branches of
\(\ket{\Psi_6}^{\otimes m}\) consequently have zero coherence in the
two-particle reduced density matrix. Since each individual branch is
a computational-basis Slater determinant,
\(\Gamma_{\mathrm{in}}^{(2)}\) is diagonal. Writing
\[
w_{pq}
:=
[\Gamma_{\mathrm{in}}^{(2)}]_{pq,pq},
\qquad
p<q,
\]
its entries are explicitly
\begin{equation}
w_{pq}
=
\begin{cases}
\frac12,
&
\substack{
p,q\text{ belong to the same triplet}\\
\text{branch of one block},
}
\\[2mm]
0,
&
\substack{
p,q\text{ belong to opposite branches}\\
\text{of one block},
}
\\[2mm]
\frac14,
&
\substack{
p,q\text{ belong to different}\\
\text{active blocks},
}
\\[2mm]
0,
&
\substack{
p\text{ or }q\text{ is an initially}\\
\text{vacant mode}.
}
\end{cases}
\label{eq:input-two-rdm-weights}
\end{equation}
Combining Corollary~\ref{cor:diagonal-r-rdm} with
\eqref{eq:input-two-rdm-weights} yields the closed formula
\begin{equation}
\boxed{
q_{ij}(\mathbf{x},\theta)
=
\sum_{p<q}
w_{pq}
\left|
W_{ip}W_{jq}-W_{iq}W_{jp}
\right|^2
}
\label{eq:explicit-two-body-algorithm}
\end{equation}
There are \(O(n^2)\) output pairs and \(O(n^2)\) terms per pair, so
the complete vector \(\mathbf q\) is computed in
\begin{equation}
\boxed{
T_{\mathbf q}(n,k)=O(n^4),
}
\label{eq:two-body-complexity}
\end{equation}
independently of \(k\), \(m=k/3\), and
\(\chi=2^{k/3}\). The same conclusion holds at inverse-polynomial
additive precision using polynomially many bits.

Gradients are also polynomially computable. If
\(R_{ij,pq}=W_{ip}W_{jq}-W_{iq}W_{jp}\), then
\[
\partial_\alpha q_{ij}
=
2\operatorname{Re}
\sum_{p<q}
w_{pq}\,
\overline{R_{ij,pq}}\,
\partial_\alpha R_{ij,pq},
\]
and \(W\) and \(\partial_\alpha W\) follow from the polynomial-size
product of known RBS and phase matrices.

For the paired form-rank-\(2\) block-product inputs studied by Oh
\emph{et al.}~\cite{Oh2026}, fixed-\(r\) reduced density matrices are
exactly classically available, so
Proposition~\ref{prop:fixed-body-closure} yields exact evaluation of
all number-conserving fixed-\(r\)-body expectations. Their mixed-Pfaffian
method instead gives additive-error estimators with complementary
scope: the reduced-density-matrix route is exact but polynomial only
for fixed \(r\), whereas their method also covers arbitrary-weight
number correlators, including weights growing with system size, and
additional overlap and transition primitives.

\section[Effect on the claims of Ref. 1]
{Effect on the claims of Ref.~\cite{Kerenidis2026v1}}

Equation~\eqref{eq:explicit-two-body-algorithm} contradicts the
algorithm-relative conclusion that evaluating the triplet-block
two-body readout requires
\(2^{\Omega(k)}\operatorname{poly}(n)\) classical time. In particular,
the \(2^{2k/3}\) Gaussian cross-term count is avoided entirely. The
form-rank-\(3\) and hyper-Pfaffian discussion concerns a more general
full-state expansion and does not obstruct propagation of the explicit
two-particle reduced density matrix. More precisely:

\begin{itemize}
\item The generic \(O(\chi^2\operatorname{poly}(n))\) simulation
procedure in Proposition~26 remains a valid upper bound, but it is not
optimal for the triplet-block two-body readout.

\item The expectation-value hardness conclusion of Theorem~28, and
the associated claim of a classically inaccessible supervised feature
vector, are invalidated by
\eqref{eq:explicit-two-body-algorithm}.

\item The supervised classification and regression models based only
on \((\langle n_i n_j\rangle)_{i<j}\), together with their parameter
gradients, admit polynomial classical evaluation.

\item The gradient-variance, barren-plateau, and parallel
parameter-shift results are not challenged.

\item The sampling results are not challenged. Output probabilities
and full sampling depend on correlations whose body order grows with
the particle number and are not determined by any fixed-particle
reduced density matrix. Supervised tasks nevertheless remain available
through sample-based heads; their relevant classical cost is the cost
of sampling rather than the cost of evaluating fixed-body
expectations.
\end{itemize}

This note gives a deterministic \(O(n^4)\) classical algorithm for the
triplet-block two-body readout. More generally, for every fixed \(r\),
number-conserving \(r\)-body expectations after passive fermionic
linear optics are polynomially computable whenever the input
\(r\)-particle reduced density matrix is classically available,
independently of the input state's form-rank; if this matrix is
diagonal, the complete diagonal correlator vector is computable in
\(O(n^{2r})\) time.

These results neither extend automatically to body order growing with
\(n\) nor affect the sampling-hardness conclusions of
Ref.~\cite{Kerenidis2026v1}. Although the reduced-density-matrix
propagation identity is standard, its application to the input and
readout of version~1 yields the fixed-body generalization and removes
the claimed exponential expectation-value cost.

\section*{Note added}

The present note concerns version~1 of
Ref.~\cite{Kerenidis2026v1}. Version~2 of that work, posted on August~20, 2026~\cite{Kerenidis2026v2}, removes the triplet encoding and all
expectation-value hardness claims and states the fixed-body
simulability result in general form.

\section*{Acknowledgments}

The author thanks Iordanis Kerenidis for constructive correspondence,
for his prompt and collegial response to the observation reported in
this note, and for suggesting the fixed-\(r\) generalization developed
in Section~\ref{sec:fixed-body-closure}.

\bibliographystyle{unsrt}
{\small
\bibliography{references}
}

\end{document}